\documentclass[12pt]{article}
\usepackage{amsmath, amsthm, amssymb}
\usepackage{indentfirst}
\usepackage{graphicx}
\usepackage{booktabs}
\usepackage{subfigure}
\usepackage{hanging}
\usepackage{multirow}
\usepackage{algorithm}
\usepackage[round, sort, authoryear]{natbib}
\usepackage{url}

\PassOptionsToPackage{normalem}{ulem}
\usepackage{ulem, xcolor}

\providecolor{added}{rgb}{0,0,1}
\providecolor{deleted}{rgb}{1,0,0}
\newcommand{\added}[1]{{\color{added}{}#1}}
\newcommand{\deleted}[1]{{\color{deleted}\sout{#1}}}

\usepackage{titlesec}
\titleformat{\section}[block]{\bfseries\filcenter}{\thesection.}{1em}{}
\titleformat{\subsection}[hang]{\bfseries}{\thesubsection.}{1em}{}

\newtheorem{assumption}{Assumption}

\newtheorem{prop}{Proposition}

\RequirePackage{amsfonts}
\newcommand{\real}{\mathbb{R}}

\providecommand{\abs}[1]{\lvert#1\rvert}

\begin{document}
\baselineskip=21pt
\def\spacingset#1{\renewcommand{\baselinestretch}%
{#1}\small\normalsize} \spacingset{1}

\title{\Large{On the Forecast Combination Puzzle}}
 \author{Wei Qian, Craig A. Rolling\thanks{Co-first author}, Gang Cheng, Yuhong Yang}
\date{}
\maketitle

\noindent


\noindent
{\bf Abstract.} It is often reported in forecast combination literature that a
simple average of candidate forecasts is more robust than
sophisticated combining methods. This phenomenon is usually referred to as the
``forecast combination puzzle''. Motivated by this
puzzle, we explore its possible explanations including
estimation error, invalid weighting formulas and model screening. We
show that existing understanding of the puzzle should be complemented by
the distinction of different  {forecast combination} scenarios known as
combining for adaptation and combining for improvement. Applying
combining methods without consideration of the underlying scenario can
itself cause the puzzle. Based on our new understandings, both simulations and real
data evaluations are conducted to illustrate the causes of the puzzle. We further propose a multi-level AFTER strategy
that can integrate the strengths of different combining methods and adapt
intelligently to the underlying
scenario. In particular, by treating the simple average as a 
candidate forecast, the proposed strategy is shown to avoid the heavy cost of
estimation error and, to a large extent, solve the forecast combination puzzle.\\


\spacingset{1.45}

\section{Introduction}

Since the seminal work of \citet{bates1969combination}, both empirical
and theoretical investigations support that when multiple candidate forecasts for a target variable are available
to an analyst, forecast combination often provides more
accurate and robust forecasting performance in terms of mean square forecast
error (MSFE) than using a single candidate forecast. The benefits of
forecast combination are attributable to the facts that individual
forecasts often use different sets of information, are subject to model bias
from different but unknown model misspecifications, and/or are 
varyingly affected by {structural} breaks. The review of \citet{timmermann2006forecast} provides a comprehensive account of various forecast combination
methods. In particular, one popular method is to combine forecasts
by estimating a theoretically optimal weight through the
minimization of  mean square error (MSE). For example, \citet{bates1969combination}
propose to find the optimal weight using error variance-covariance structure of
the individual forecasts. \citet{granger1984improved}  
construct the optimal weight under a linear
regression framework.

Despite the
ever-increasing popularity and sophistication of combining methods,
it is repeatedly
reported from past literature that the simple
average (SA) is a very effective and robust forecast combination method that
often outperforms more complicated combining methods  (see
\citet{winkler1983combination}, \citet{clemen1986combining} and
\citet{diebold1990use} for some early examples). In a review and
annotated bibliography on earlier studies, \citet{clemen1989combining}
raises the question, ``What is the explanation for the robustness
of the simple average of forecasts?''. Specifically, he proposes two questions of
interest, ``(1) Why does {the simple average} work so well, and (2) under what conditions do
other specific methods work better?'' The robustness of SA is also echoed in more recent literature. For example,
\citet{stock2004combination} build autoregressive models with
univariate predictors (macroeconomic variables) as candidate
forecasts for output growth of seven developed countries, and find
that SA, together with other methods of least data adaptivity,
is among the top-performing forecast combination methods.
\citet{stock2004combination} further coin the term ``{\bf F}orecast
{\bf C}ombination
{\bf P}uzzle'' (for brevity, we refer to the puzzle as FCP {hereafter} ), which refers to ``the repeated finding that simple
combination forecasts outperform sophisticated adaptive combination
methods in empirical applications''. In another recent example,  \citet{genre2013combining} use survey data from professional
forecasters as the individual candidates to construct combined
forecasts for three target variables. Despite some promising results
of  complicated methods, they further note that
the observed improvement over SA is rather vague when a period
 of financial crisis is included in the analysis. The past empirical evidence
appears to support the mysterious existence of FCP, which is also  summarized in \citet[section 7.1]{timmermann2006forecast}.

Many attempts have been made to demystify FCP.  One popular and arguably the most well-studied explanation for FCP
is the estimation error of the combining methods that
rely on the optimal weight estimation by MSE minimization.  \citet{smith2009simple} 
rigorously study the estimation error issue. Using the forecast error
variance-covariance structure, they show both theoretically and
numerically that  the estimator targeting the optimal weight can have large variance and
consequently, the estimated optimal weight can be very different from
the true optimal weight, often even more so than simple equal
weight. \citet{elliott2011averaging} studies the theoretical maximal
performance gain of the optimal weight over SA by optimizing the error
variance-covariance structure, and points out that the gain is often
small enough to be overshadowed by estimation error. \citet{timmermann2006forecast}
and \citet{hsiao2014there} also illustrate conditions for the optimal weight
to be close to the equal weight so that the relative gain of the optimal
weight over SA is small.  \citet{claeskens2014forecast} consider the random weight and show that
when the weight variance is taken into account, SA can perform better than
using the ``optimal'' weight. Under linear
regression settings, \citet{huang2010combine} discuss the estimation
error and the relative gain of the optimal weight. 

In addition to
estimation error, nonstationarity and structural breaks in {the} data generating process (DGP) are believed to contribute to the
unstable performance of the estimated ``optimal'' weight. For example,
\citet{hendry2004pooling} demonstrate that when candidate forecasting
models are all misspecified and breaks
occur in the information variables, forecast combination methods that target the
optimal weight may not perform as well as SA. Also,
\citet{huang2010combine} propose that the candidate forecasts
are often weak, that is, they have low predictive content on the
target variable, making the optimal weight similar to simple equal weight. 

While the aforementioned points are valid and valuable, they do not depict the complete picture of the puzzle. In this paper, we provide our perspectives on FCP to contribute to its settling. In  our view, besides providing explanations of FCP, it is also very important to point out the potential danger of recommending SA for broad and indiscriminate use. Here, we focus on the mean squared error (MSE). It should be pointed out that the main points are expected to stand for other losses as well (e.g., absolute error) and {that} some combination approaches (e.g., AFTER) can handle general loss functions.  

The rest of this article is organized as follows. In section 2, we list some aspects that have not been much addressed but are important towards the understanding of FCP in our view. We formally introduce
the problem setup of the forecast combination problem we consider in
section~\ref{sec:setup}.  Our understandings of FCP are
elaborated in sections~\ref{sec:two_goals}-\ref{subsec:screening}.  {In particular, section~\ref{sec:cocktail0} proposes a multi-level AFTER approach to solve FCP.} The
performance of this approach is also evaluated in section~\ref{sec:real_data}
using a U.S. Survey of Professional Forecasters (SPF) data. A brief conclusion is given in section~\ref{sec:conclusion}.

\section{Additional Aspects of FCP}

The previous work has nicely pointed out that
estimation error is an important source of FCP and has characterized the impact of the estimation error in idealized settings. Indeed, in general, when the
forecast combination weighting formula is valid in the sense that an
optimal weight can be correctly estimated by minimizing MSE, insufficiently small sample size  may not support
reliable estimation of the weight, resulting in inflated variance of
the combined forecast. The explanation with {structural} breaks also makes sense for certain situations. However, in our view,
there are several  additional aspects that need to be considered for understanding FCP.

\begin{enumerate}

\item A key factor missing in addressing the FCP is the true nature of improvability of the candidate forecasts. While we all strive for better forecast performance than the candidates, that may not be feasible (at least for the methods considered). Thus we have two scenarios \citep{yang2004combining}: i) One of the candidates is pretty much the best we can hope for (within the considerations of course) and consequently any attempt to beat it will not succeed. We refer to this
scenario as ``Combining for Adaptation'' (CFA), because the proper goal of a forecast combination method under this
scenario should be targeting the performance of the best individual candidate
forecast, which is unknown. ii) The other is that a significant gain of accuracy over all the individual candidates can be materialized.  
 We refer to this scenario
as ``Combining for Improvement'' (CFI),
because the proper goal of a forecast combination method under this
scenario should be targeting
the performance of the best combination of the candidate
forecasts to overcome defects of the candidates. In our experience, both scenarios occur commonly in real problems. Without factoring in this aspect, comparison of different combination methods may be grossly misleading due to the well-known sin of comparing apples to oranges.  In our view, empirical studies on forecast combinations in the future need to bring this lurking aspect into the analysis. With the above forecast combination scenarios spelled out, a natural question {follows}: Can we design a combination method to bridge the two camps of methods proposed for the two scenarios respectively{,} so as to help solve the FCP? 

\item  The methods being examined in the literature on FCP are mostly specific choices (e.g., least squares estimation).  Can we do better with other methods ({that} may or may not have been invented {yet}) to avoid the heavy estimation price? Also, the currently investigated methods often assume the forecasts are unbiased and the forecast errors are stationary, which may not be proper for many applications. What happens when {these assumptions} fail?

\item It has been stated in the literature that the simple methods (e.g., SA) are robust based on empirical {studies}. We feel this is not necessarily true in the usual statistical sense (rigorously or loosely). In many published empirical results, the candidate forecasts were carefully selected/built and thus well-{behaved}.  Therefore, the finding in favor of robustness of SA may be proper only for such situations that the data analyst has extensive expertise on the forecasting problem and has done quite a bit { of }work on screening out poor/un-useful candidates. We argue that it is much more desirable to investigate FCP broadly so as to allow the possibility of poor/redundant candidates for wider and more realistic applications. It should be added that in various situations, the screening of forecasts is far from an easy task and its complexity may well be at the same level as model selection/averaging. Therefore, even for top experts, the view that we can do a good job in screening the candidate forecasts and then simply recruit SA is overly optimistic. With the above,  an important matter is to examine the robustness of SA in a broader context.   

\end{enumerate}

As is described in the first item, there are two distinct scenarios:
CFA and CFI.  The CFA scenario can happen if one of the candidate forecasts
is based on a model
sophisticated enough to capture the true DGP (yet still relatively simple), and/or the other candidate
forecasts only add redundant information. The CFI scenario can often
happen when different candidate forecasts use different information,
and/or their underlying models have misspecifications in different ways.   

There are different existing combining methods designed for the two
scenarios. The methods for the CFI scenario 
typically seek to estimate the optimal weight aggressively,
and their examples include
variance-covariance based optimization \citep{bates1969combination} and
linear regression \citep{granger1984improved}. These methods are likely to suffer from estimation error, causing
unstable performance relative to SA. On the other hand, the combining
methods for the CFA scenario should ideally
perform similarly to the best individual candidate forecast and
should not be subject as severely to 
estimation error as the methods for CFI. The
typical methods suitable for
the CFA scenario include AIC model averaging \citep{buckland1997model} and Bayesian model
averaging (e.g., \citealp{garratt2003forecast}), both in parametric settings. {The method of AFTER \citep{yang2004combining} can be applied more broadly in parametric and non-parametric settings, regardless of the nature of the candidate forecasts.} As one of the main contributions in this article, we show that the distinction between the two scenarios provides one of
the keys to understanding the FCP. We will see in
section~\ref{sec:two_goals} that an analyst who fails to understand and bring in the
underlying scenarios and specific types of data when choosing the
combining methods can incorrectly
apply a combining method not designed for
the underlying scenario and consequently deliver forecasting results
worse than other methods (e.g., SA). 

For the questions raised in the second item regarding whether we can avoid  the estimation price, we cannot fully
address them without  a proper
framework, because for any sensible method, one can always find 
a situation to favor it to its competitors. The framework we consider
with sound theoretical support is through a minimax view: If one has a
specific class of combination of the forecasts in mind and wants to
target the best combination in this class, then without any
restriction/assumption on
unbiasedness of the candidate forecasts and stationarity of the
forecast errors, the minimax view seeks a clear understanding of the
minimum price we have to pay no matter what method (existing or not)
is used for combining. It turns out that the framework from the minimax
view is closely related to the forecast combination scenarios{ discussed in the first item}, and
\citet{yang2004combining} provides a detailed theoretical exposition of the
distinct forecast combination scenarios and associated minimax results. 

Indeed, \citet{yang2004combining} shows that  from a minimax perspective,
because of the aggressive target set for
the CFI scenario, we have to pay an unavoidably
heavier cost than the target set under the
CFA scenario. Specifically,  if we let $K$ denote the number
of forecasts and $T$ denote the forecasting horizon,
\citet{yang2004combining} shows that when the target is to find the
optimal weight to minimize the general empirical risk over a set
of weights satisfying a convex constraint (which is appropriate under
the CFI scenario), the estimation cost is
$O(\frac{K \log(1+T/K)}{T})$ for relatively large $T$ ($T>K^2$), and $O(\log (K)/\sqrt{T\log
  T})$ for relatively small $T$ ($T\leq K^2$).  In
contrast, if the target is to match the performance of the best
individual forecast (which is appropriate under the CFA
scenario), the estimation cost is only  $O(\log (K)/T)$. 

Because of the unavoidable heavy cost under the CFI scenario, it is not
always ideal to pursue the 
aggressive target of the optimal weight. Indeed, even if the optimal weight gives better performance than the best individual candidate, the improvement may not be enough to offset the additional estimation cost (i.e., increased variance) as precisely (in minimax rate) identified in \citet{yang2004combining} and \citet{wang2014adaptive}.  As another contribution of
our work, we show in section~\ref{sec:cocktail}
that an appropriately constructed  forecast
combination strategy can perform in a smart way according to  the
underlying CFI or CFA scenario. If CFI is the correct scenario, the proposed 
strategy can behave both aggressively and
conservatively so that it performs similar to SA when SA is much better
than e.g., the linear regression method.

Besides the estimation error and the necessary distinction of
underlying scenarios discussed in the first two items, the following three reasons
can also contribute to FCP. First, the weighting derivation formula used by complicated methods is often
not suitable for the situation. For example, under {structural} breaks, old
historical data no longer hold support for a valid optimal weighting scheme, and the
known justification of
well-established combining methods fails as a result. Indeed,
 \citet{hendry2004pooling} demonstrate that when candidate forecasting
models are all misspecified and breaks
occur in the information variables, methods that estimate the
optimal weight may not perform as well as SA. In
section~\ref{subsec:breaks}, our Monte Carlo examples also show that SA may
dominate the complicated methods when breaks occur in DGP dynamics.  Second, it is common practice that the candidate forecasts are already screened in some
ways so that they are more or less on an equal footing. For example,
\citet{stock1998comparison} and \citet{stock2004combination} apply
various model selection methods such as AIC and BIC to identify promising
linear or nonlinear candidate forecast models. Recently,
\citet{bordignon2013combining} select models of different types
(ARMAX, time-varying coefficients, etc.) and suggest that SA works
well when combining a small number of well-performing forecasts. In studies using survey data of professional forecasters, it is also
expected that each professional forecaster performs some model screening
before satisfactorily settling down with their own forecast.  In these cases,
there may not be particularly poor candidate forecasts, and the
the candidates (at least the top ones) may tend to contribute more or less equally to the optimal combination, making SA a 
competitive method.  In section~\ref{subsec:screening}, we use
Monte Carlo examples to show that screening can be a source of FCP.  Lastly, the puzzle can also be
a result of publication bias; people do not tend to emphasize the
performance of SA when SA does not work well. 

With all our understandings of FCP discussed above, we address the
issues raised in the third item and provide further information
on robustness of SA in
sections~\ref{sec:cocktail}-\ref{subsec:screening}. In particular, we
will see that  SA is actually not robust in performance in several directions: its performance may change significantly or even substantially when i) an optimal, poor or redundant forecast is added; or ii) the degree of the screening of the candidate forecasts is done differently. In addition, the size of the rolling window to deal with {structural} breaks affects the relative performance of SA as well. Fortunately, as will be seen, some combination methods can{ largely }avoid these defects.


\section{Problem Setup}
\label{sec:setup}

Suppose that an analyst is interested in forecasting a real-valued time series
$y_1,y_2,\cdots$. Given each time point $t\geq 1$, let $\mathbf x_t$
be the (possibly multivariate) information variable vector revealed
prior to the observation of $y_t$. {The }$\mathbf x_t$ may not be accessible to
the analyst. Conditional on $\mathbf x_t$ and $\mathbf z_{t-1} =:\{(\mathbf x_j,
y_j),\,1\leq j\leq t-1\}$, $y_t$ is subsequently generated from some unknown
distribution $p_t(\cdot|\mathbf x_t, \mathbf z_{t-1})$ with conditional
mean $m_t=\mathbb E(y_t|\mathbf x_t, \mathbf z_{t-1})$ and conditional
variance $v_t=\text{Var}(y_t|\mathbf x_t, \mathbf z_{t-1})$. Then,  $y_t$ can
be represented as
$y_t=m_t+\varepsilon_t$, where $\varepsilon_t$ is the random noise
with the conditional mean and the conditional variance being 0 and
$v_t$, respectively. 

Assume that prior to the observation of
$y_t$, the analyst has access to $K$ real-valued candidate forecasts $\hat
y_{t,i}$ ($i=1,\cdots,K$). These forecasts may be
constructed with different model structures, and/or with different
components of the information variables, but the details
regarding how each original forecast is created may not be
available in practice and are not {assumed} to be known. {The analyst's objective in (linear) forecast combination }is to construct a weight vector $\mathbf
w=(w_{1},\cdots,w_{K})^T\in \mathbb R^K$, based on the available information
prior to the observation of $y_t$,
 to find a point forecast
of $y_t$ by forecast combination  $\hat y_{t,\mathbf
  w}=\sum_{i=1}^K w_{i}\hat y_{t,i}$.  The weight vector may
be different at different time points.

To gauge the performance of a {procedure that produces forecasts $\{\hat{y}_t, t=1,2,\dots\}$} given 
time horizon $T$, we consider the average forecast risk
\begin{equation*}
R_T = \frac{1}{T}\sum_{t=1}^T\mathbb E(y_t-\hat y_t)^2
\end{equation*}
in our analysis and simulation studies. For real data
evaluation, since the risk cannot be computed, we use the mean square
forecast error (MSFE) as a substitute:
\begin{equation*}
\text{MSFE}{_T}=\frac{1}{T}\sum_{t=1}^T(y_t-\hat y_t)^2.
\end{equation*}

According to{ the }FCP, simple methods with little or no time variation in weight $\mathbf
w$ (e.g., equal weighting) often outperform complicated methods with much time variation in
terms of {$R_T$ and $\text{MSFE}_T$}.



\section{CFA versus CFI: A Hidden Source of FCP}
\label{sec:two_goals}

In this section, we study the performance of forecast
combination methods under the two distinct scenarios. Failure to
recognize these scenarios can itself result in{ the }FCP.  We use two simple but illustrative Monte Carlo examples under regression settings similar to{ those }of \citet{huang2010combine} to demonstrate the CFA and
CFI scenarios. 

\begin{description}
\item[Case 1.] Suppose $y_t$ ($t=1,\cdots, T$) is generated by the linear model
\begin{equation*}
y_t=x_t\beta+\varepsilon_t,
\end{equation*}  
where $x_t$'s are  $i.i.d.$ $N(0,\sigma_X^2)$, and $\varepsilon_t$'s are independent of
$x_t$'s and are  $i.i.d.$ $N(0,\sigma^2)$. Consider the
two candidate forecasts generated by 
\begin{align*}
\text{Forecast 1: }& \hat y_{t,1} = x_t\hat\beta_t;\\
\text{Forecast 2: }& \hat y_{t,2} = \hat\alpha_t,
\end{align*}
where $\hat\beta_t$ and $\hat\alpha_t$ are both obtained from the
ordinary least square (OLS) estimation using historical data.
\end{description}
Given that Forecast 1 essentially represents the true model, 
its combining with Forecast 2 cannot improve over the performance of
the best individual forecast asymptotically, thus giving an example of the CFA scenario. Let
$T_0$ be a {\it fixed} start point of the evaluation period, and
let $T$ be the end point. Given the evaluation period from $T_0$
to $T$, let $R_{T,1}$, $R_{T,2}$ and $R_{T,\mathbf w}$ be
the average forecast risks of Forecast 1, Forecast 2 and the
combined forecast, respectively. If we let $R_{T,\text{SA}}$
be the average
forecast risk at time $T$ for SA, we expect that
$R_{T,\text{SA}}>R_{T,1}$. Indeed, Proposition~\ref{prop:case1} in the Appendix shows
\begin{equation}
\label{eq:ratio1}
\frac{R_{T,1}}{R_{T,\text{SA}}} \rightarrow
\frac{\sigma^2}{\sigma^2+\beta^2\sigma_X^2/4} \quad \text{as }\,
T\rightarrow \infty,
\end{equation}
and asymptotically, the optimal combination assigns all the weight on 
Forecast 1.

Under the CFA scenario, since the best candidate is unknown,
the natural goal of forecast combination is to match the performance
of the best candidate. 

\begin{description}
\item[Case 2.] Suppose $y_t$ ($t=1,\cdots, T$) is generated by the
linear model
\begin{equation*}
y_t=\left(x_{t,1}+x_{t,2}\right)\beta+\varepsilon_t,
\end{equation*}
where the $\mathbf x_{t}=(x_{t,1},x_{t,2})^T$ {are }$i.i.d.${ following a }bivariate normal distribution with mean $\mathbf 0$  and common variance $\sigma^2_X = \sigma^2_{X_1} = \sigma^2_{X_2}$. Let $\rho$ denote the correlation between $x_{t,1}$ and $x_{t,2}$. The random error  $\varepsilon_t$'s are independent of
$\mathbf x_t$'s and are  $i.i.d.$ $N(0,\sigma^2)$. Consider the two
candidate forecasts generated by
\begin{align*}
\text{Forecast 1: }& \hat y_{t,1} = x_{t,1}\hat\beta_{t,1};\\
\text{Forecast 2: }& \hat y_{t,2} = x_{t,2}\hat\beta_{t,2},
\end{align*}
where $\hat\beta_{t,1}$ and $\hat\beta_{t,2}$ are both obtained from
OLS estimation with historical data.
\end{description}
Different from Case 1, Case 2  presents a scenario where each candidate forecast employs{ only }part
of the information set. It is expected, to some extent, that combining the two
forecasts works like pooling different sources of important information, resulting in performance
better than either of the candidate forecasts. By defining the average forecast
risks $R_{T,1}, R_{T,2}, R_{T,\text{SA}}$ the same
way as in Case 1, we can see from Proposition~\ref{prop:case2} in the Appendix that
\begin{equation}
\label{eq:ratio2}
\frac{R_{T,1}}{R_{T,\text{SA}}} \rightarrow
\frac{\sigma_X^2\beta^2(1-\rho^2)+\sigma^2}{\sigma_X^2\beta^2(1-\rho^2)(1-\rho)/2+\sigma^2} \quad \text{as }\,
T\rightarrow \infty.
\end{equation}
Clearly, when the two information sets are not highly correlated, SA can improve the forecast performance over
the best candidate. This case gives a typical example of the
CFI scenario, and it is appropriate to seek the more aggressive
goal of finding the best linear combination of candidate forecasts. 
\par
{Our view is that }discussion of{ the }FCP should
take into account the{ different combining }scenarios.  Next, we perform Monte Carlo studies on the two cases to provide an explanation of the
puzzle. Combining methods suitable for the CFA scenario have
been developed to target performance of the best individual
candidate. In our numerical studies, we choose{ the }AFTER method \citep{yang2004combining} as the
representative, and it is known that AFTER pays a smaller estimation price than
methods that target the optimal linear or convex weighting.  In contrast, combining
methods for the CFI scenario usually attempt to
estimate the optimal weight. We choose linear regression{ of the response on the candidate forecasts }(LinReg) as the representative. The method of \citet{bates1969combination} without estimating correlation (BG for brevity) is used as{ an additional }benchmark. 

For Case 1, we perform simulations as follows. Set $\sigma^2=\sigma_X^2=1$. Consider a sequence of 20 $\beta$'s
such that the corresponding  signal-to-noise (S/N) ratios are evenly spaced between 0.05
and 5 in the logarithmic scale. For each $\beta$, we conduct the following
simulation 100 times to {estimate }the average forecast risk. A sample of
100 observations is generated. The first 60 observations are used
to build the candidate forecast models, which are subsequently used to generate forecasts for
the remaining 40 observations. Forecast
combination methods including SA, BG, AFTER and LinReg methods are applied to
combine the candidate forecasts, and the last 20 observations are used for
performance evaluation. The average forecast risk of each forecast combination method
is divided by that of SA to obtain the normalized average forecast
risk (denoted by normalized $R_T$). The
results are summarized in Figure~\ref{fig:toy1}.
\begin{figure}[!ht]
\centering
\includegraphics[scale=1]{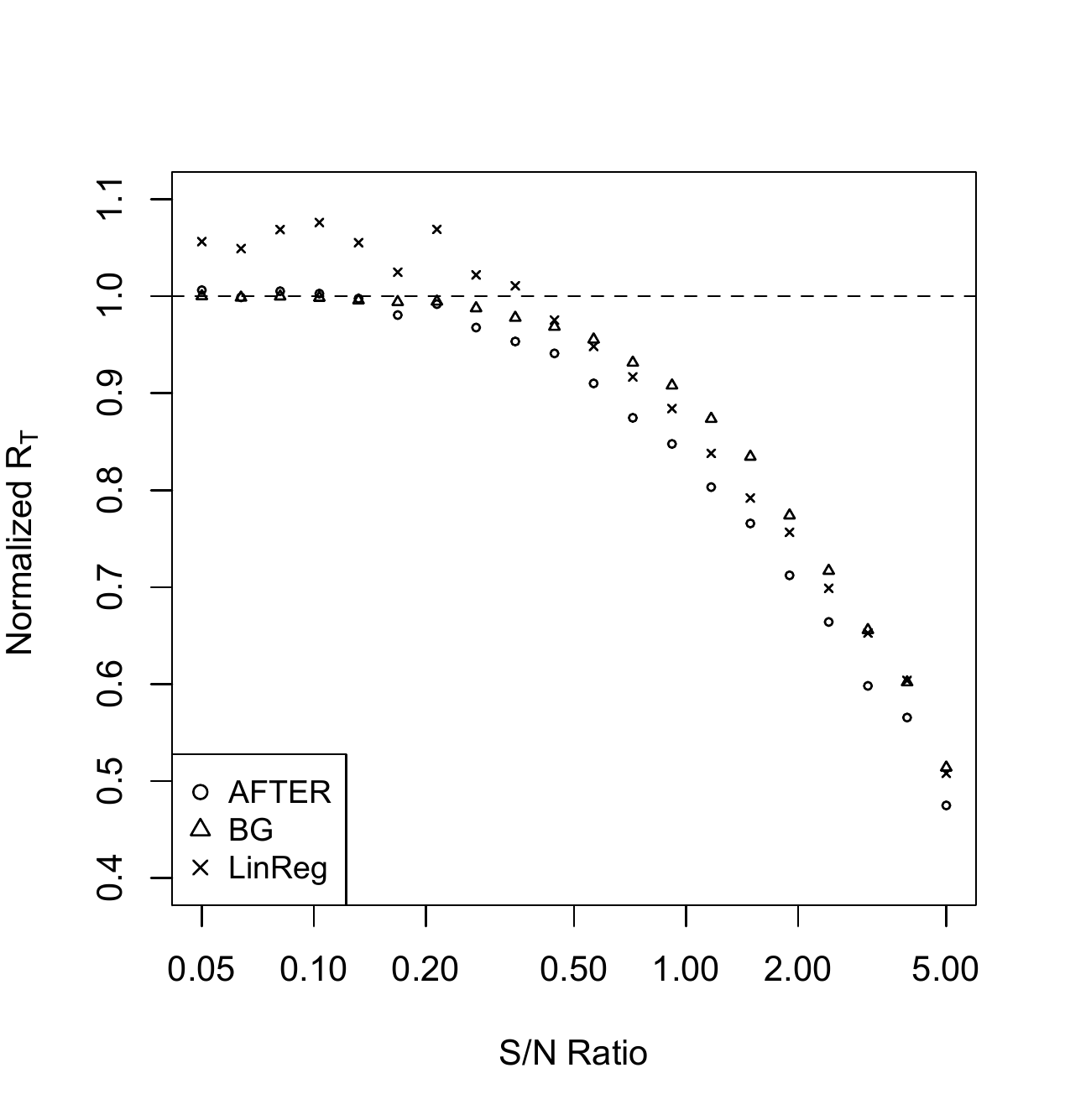}
\caption{(Case 1) Comparing the average forecast risk of different forecast
  combination methods (dashed line represents the SA baseline; 
  $x$-axis is in logarithmic scale).}
\label{fig:toy1}
\end{figure}
For Case
2, we set $\beta=\beta_1=\beta_2$, $\rho=0$ and
$\sigma^2=\sigma_{X_1}^2=\sigma_{X_2}^2=1$. The remaining simulation settings
are the same as Case 1. The normalized average forecast risks (relative to
SA) are summarized in Figure~\ref{fig:toy2}.

\begin{figure}[!ht]
\centering
\includegraphics[scale=1]{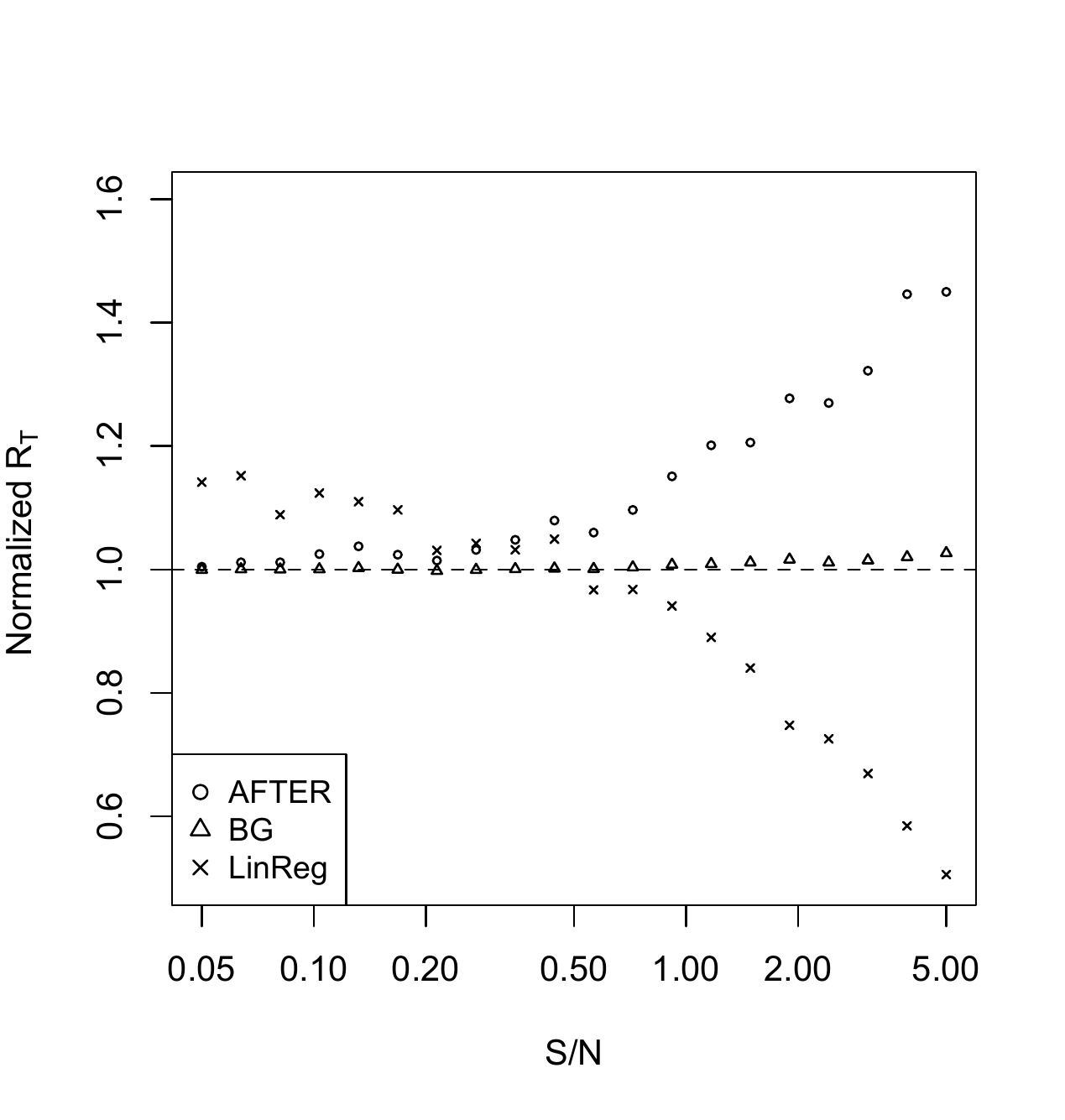}
\caption{(Case 2) Comparing the average forecast risk of different forecast
  combination methods (dashed line represents the SA baseline; 
  $x$-axis is in logarithmic scale).}
\label{fig:toy2}
\end{figure}

In Case 1, it is clear from Figure~\ref{fig:toy1} that AFTER is the preferred
method of choice under the CFA scenario. LinReg, on the
other hand, consistently underperforms compared to
AFTER. Interestingly, when S/N is relatively low (less than 0.35), we observe the
``puzzle'' that LinReg 
performs worse than SA, which is due to the weight
estimation error. If the analyst correctly identifies that it is the
CFA scenario and applies a corresponding method like AFTER,
the ``puzzle'' disappears: AFTER can perform better than
(or very close to) SA, 
while LinReg fails.

In Case 2, if the analyst applies AFTER
without realizing the underlying CFI scenario, we observe the
``puzzle'' that SA outperforms AFTER. The ``puzzle'' is
not entirely surprising since AFTER is designed to target the
performance of the best individual
forecast, while \eqref{eq:ratio2} shows that SA can improve over the
best individual forecast. LinReg appears to be the correct
method of choice when S/N ratio is relatively high. However, similar
to what is observed in Case 1,
LinReg suffers from weight estimation error when S/N ratio
is low,  once again giving the ``puzzle'' that LinReg performs worse than
SA.

Case 2 also shows the interesting observation that it is not
always optimal to apply SA even when SA is the ``optimal'' weight in a
restricted sense. Indeed, \eqref{eq:weight_improve} and
\eqref{eq:weight_improve1} in Proposition~\ref{prop:case2} imply
that if we adopt the
common restriction that the sum of  all weights is 1, SA is
the asymptotic optimal weight. However, if we impose no restriction on
the weight range, the asymptotic optimal weight assigns a unit weight to
each candidate forecast. This explains the advantage of LinReg over SA in Case 2 when the S/N ratio is large.

The observations above illustrate that different combining methods can
have strikingly different performance depending on the underlying
scenario. The FCP can appear when a combining method is not properly
chosen according to the correct scenario. Without knowing the underlying scenario, comparing these
methods may not provide a complete picture of FCP, and blindly
applying SA may result in sub-optimal performance. We advocate the
practice of trying to  identify the underlying scenario  (CFA or CFI) when considering forecast combination. It should be pointed out
that when the relevant information is limited, it may not be feasible
to confidently identify the forecast combination scenario. In such a
case, a forced selection, similar to the comparison of model selection
and model combining (averaging) described in \citet{yuan2005combining}, would induce enlarged variability of the resulting forecast. A better solution is an adaptive combination of forecasts as {illustrated }in the next section.  

\section{Multi-level AFTER}
\label{sec:cocktail0}
With the understanding in section~\ref{sec:two_goals},{ we see that }when considering forecast combination methods, an effort{ should }be made to understand whether there is much room for improvement over the best candidate. When this is difficult to decide or impractical to implement due to handling a large number of quantities to be forecast in real time, we may turn to the question: Can we find an {\it adaptive} (or {\it universal}) combining strategy that 
performs well in both CFA and CFI scenarios? Note that here {\it
  adaptive} refers to adaptation to the forecast combination scenario
(instead of adaptation to achieving the best individual
performance).  Another question {follows}:  Under the CFI scenario, can the adaptive combining strategy still perform as well as SA when the price of estimation error
is high? As we have seen in Case 2 of section~\ref{sec:two_goals}, using methods (e.g., LinReg)  intended for the CFI scenario alone cannot successfully address the second question.

It turns out that the answers to these two questions are
affirmative. The idea is related to a philosophical
comment in \citet{clemen1995screening}: \\
{\it ``Any combination of forecasts yields a single forecast. As a
  result, a particular combination of a given set of forecasts can
  itself be thought of as a forecasting method that could compete...''}\\
The use of combination of forecast (or procedure) combinations is a theoretically powerful tool to achieve adaptive minimax optimality (see, e.g., \citet{yang2004combining}, \citet{wang2014adaptive}). In the context of our discussion, combined forecasts such as SA, AFTER and LinReg can
all be considered as the candidate forecasts and may be used as individual candidates in a forecast combination scheme. 

Accordingly, we design a two-step
combining strategy: first, we construct three new candidate forecasts using  SA, AFTER and
LinReg; second, we apply the AFTER algorithm on
these new candidate forecasts to generate a combined forecast. We refer to this two-step algorithm as multi-level
AFTER (or mAFTER for short) because two layers of AFTER algorithms are
involved. The key lies in the 
AFTER algorithm on the second step, which allows mAFTER to
automatically target the
performance of the best individual candidate among SA, AFTER and LinReg. Under 
the CFA scenario, mAFTER can perform as if we are using AFTER alone
considering that AFTER is the proper method of choice. Under the CFI scenario, mAFTER can perform closely
to the better of SA and LinReg. Thus,  when LinReg suffers from severe estimation error, mAFTER will perform closely to SA and thereby avoid the high cost.

Indeed, if we denote the forecasts generated from SA, LinReg
and mAFTER by $\hat y_t^{({SA})}$, $\hat y_t^{({LR})}$ and $\hat
y_t^{({M})}$, respectively, we have Proposition~\ref{cor:cocktail} as follows.
\begin{prop}
\label{cor:cocktail}
Under the regularity conditions shown in the Appendix, the
average forecast risk of the mAFTER strategy satisfies
\begin{align*}
\frac{1}{T}\sum_{t=T_0}^T\mathbb E(y_{t}-\hat y_{t}^{(M)})^2
\leq & \inf\Bigl(\inf_{1\leq i\leq K}\frac{1}{T}\sum_{t=T_0}^T\mathbb
E(y_{t}-\hat y_{t,i})^2 +\frac{c_1\log(K)}{T},\,
\frac{1}{T}\sum_{t=T_0}^T\mathbb E(y_{t}-\hat
y_t^{({SA})})^2+\frac{c_2}{T},\,\\
&\qquad \frac{1}{T}\sum_{t=T_0}^T\mathbb E(y_{t}-\hat
y_t^{({LR})})^2+\frac{c_2}{T}\Bigr),
\end{align*}
where $c_1$ and $c_2$ are some positive constants not depending on the
time horizon $T$.
\end{prop}

Proposition~\ref{cor:cocktail} is a consequence of Theorem 5 in
\citet{yang2004combining}. It indicates that, in
terms of the average forecast risk, mAFTER can match
the performance of the best original individual forecast, the SA
forecast and the LinReg forecast (whichever is the best), with a relatively small price of order{ at most }$\log(K)/T$.

To confirm that the mAFTER strategy can solve the ``puzzles'' 
illustrated in the previous section, we repeat the simulation studies
of Case 1 and Case 2 and summarize the results
in Figure~\ref{fig:toy_cocktail1} and Figure~\ref{fig:toy_cocktail2},
respectively. In Case 1, it suffices to see that mAFTER correctly tracks the performance of AFTER. In Case 2, when S/N is
relatively large ($>0.5$), mAFTER takes advantage of the opportunity
to improve over the original individual forecasts and performs very
closely to LinReg; when S/N is relatively small ($< 0.5$), mAFTER
behaves very similarly to SA and successfully avoids the heavy
estimation error suffered by LinReg. Therefore, rather than relying on
SA, a ``sophisticated'' combining strategy like mAFTER can be an appealingly
safe method that avoids FCP.

Note that mAFTER is a rather general forecast combination
strategy. In the first step of the strategy, the analyst can choose
their own way of generating new candidate forecasts (not necessarily
restricted to AFTER and LinReg), as long as they
include SA, representative methods for the CFA scenario, and representative
methods for
the CFI scenario. AFTER and LinReg are simply chosen in our study
as convenient representatives.  We also demonstrate the performance of
the mAFTER strategy in the real
data example in section~\ref{sec:real_data}.

\begin{figure}[!ht]
\centering
\includegraphics[scale=1]{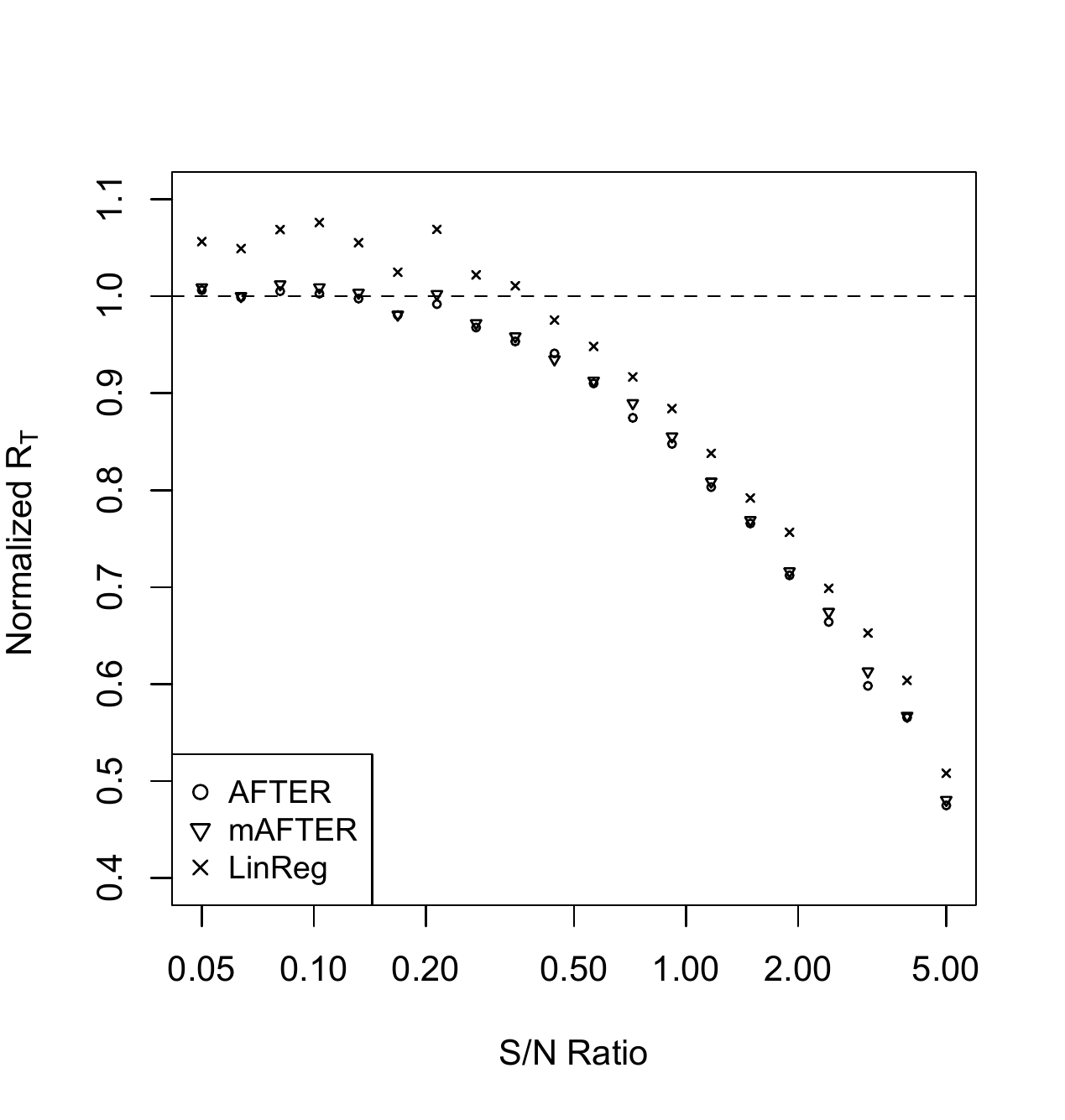}
\caption{(Case 1) Performance of mAFTER under adaptation scenario (dashed line represents the SA baseline; 
  $x$-axis is in logarithmic scale).}
\label{fig:toy_cocktail1}
\end{figure}

\begin{figure}[!ht]
\centering
\includegraphics[scale=1]{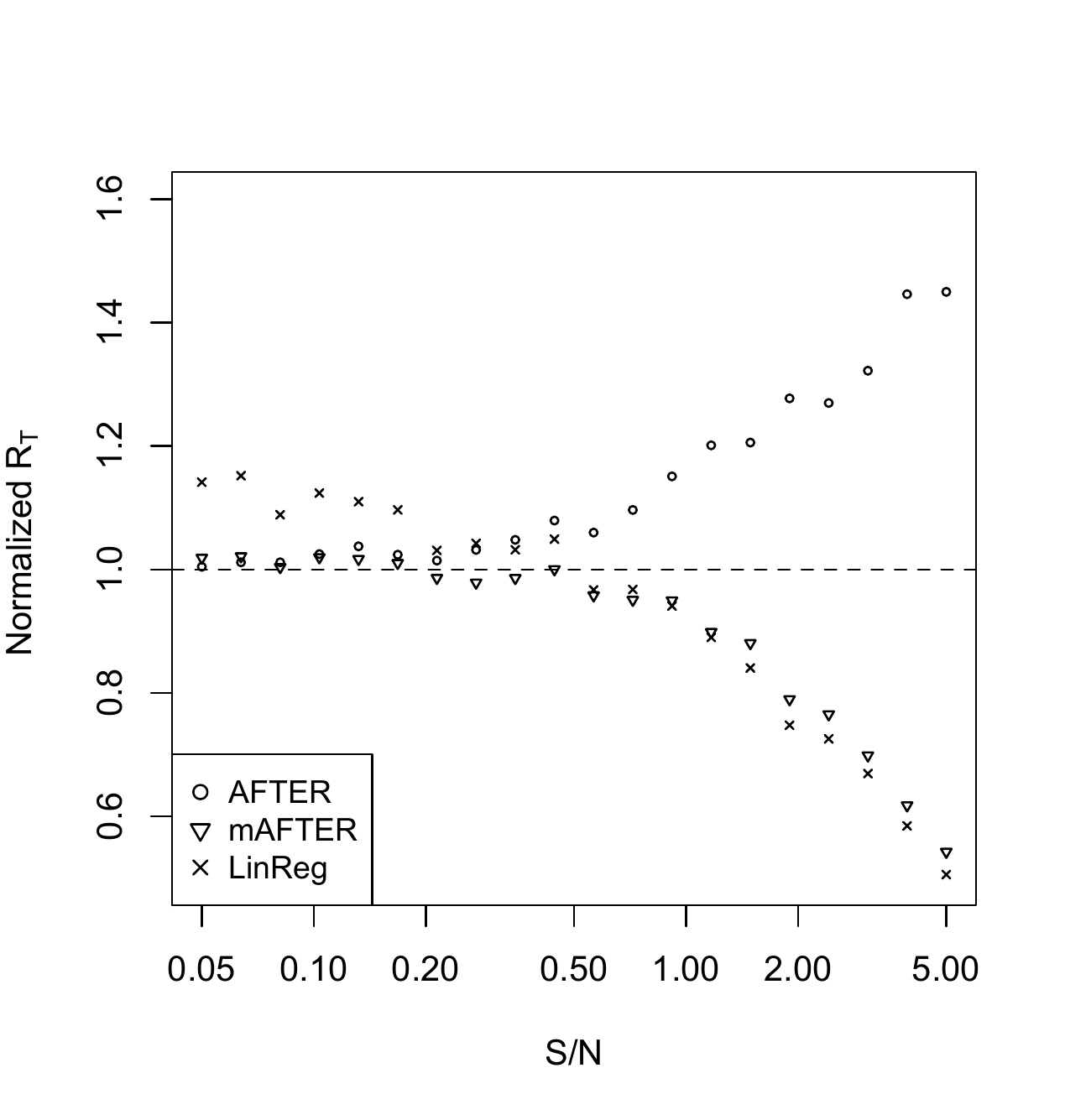}
\caption{(Case 2) Performance of mAFTER under improvement scenario (dashed line represents the SA baseline; 
  $x$-axis is in logarithmic scale).}
\label{fig:toy_cocktail2}
\end{figure}

\section{Is SA Really Robust?}
\label{sec:cocktail}

The SA has been praised for being robustly among top performers {relative to other }forecast combination methods. It is obvious that SA cannot be robust in the traditional statistical sense: even a single really bad candidate can damage the performance of the combined forecast to an arbitrarily worse position. 
A more interesting question is to assess robustness of SA in practically relevant settings.

The previous two sections have shown that SA is not robust  in terms of its relative performance when dealing with the
two different scenarios. In this section, we show that SA is not robust even in the loose sense when 
 new forecast candidates are
added to the candidate pool, especially if the new candidates have
only redundant information with respect to the original candidate pool. In contrast, the AFTER-type combining methods can be
rather robust against adding poor or redundant candidate
forecasts. Here, we
consider the
following three cases. 

\begin{description}
\item[Case 3.] Suppose a new information variable $x_{t,3}$ has the
  same distribution as $x_{t,1}$, and is
  independent of $\mathbf z_{t-1}$ and $(x_{t,1},x_{t,2})$. A new candidate forecast 
$\hat y_{t,3} = x_{t,3}\hat\beta_{t,3}$
joins the candidate pool in Case 2, where $\hat\beta_{t,3}$ is
obtained from OLS estimation with historical data.

\item[Case 4.] A new candidate forecast 
$\hat y_{t,3}=x_{t,2}\hat\beta_{t,2}$ identical to Forecast 2 
joins the candidate pool in Case 2. 

\item[Case 5.] A new candidate forecast $\hat y_{t,3}=\tilde x_{t,2}\tilde\beta_{t,2}$ is generated using a
  transformed information variable $\tilde
  x_{t,2}=\exp(x_{t,2})$, where $\tilde\beta_{t,2}$ is obtained from
OLS estimation with historical data.

\end{description}

Note that the new candidate in Case 3 is a very poor forecast, while
the new candidates in Case 4 and Case 5 contain a subset of the
information variables. In all {of }the cases above, no new information is added to
the candidate pool. Following the same simulation setting as Case 2,
we focus on SA and AFTER and compute the ratio between the MSFE after adding the new candidate and
the MSFE in Case 2. Figure~\ref{fig:toy3_5} shows
that the performance of AFTER remains almost the same, while the performance of
SA worsens after adding the {non-informative or }redundant candidate forecasts.

\begin{figure}[!ht]
\centering
\includegraphics[scale=.8]{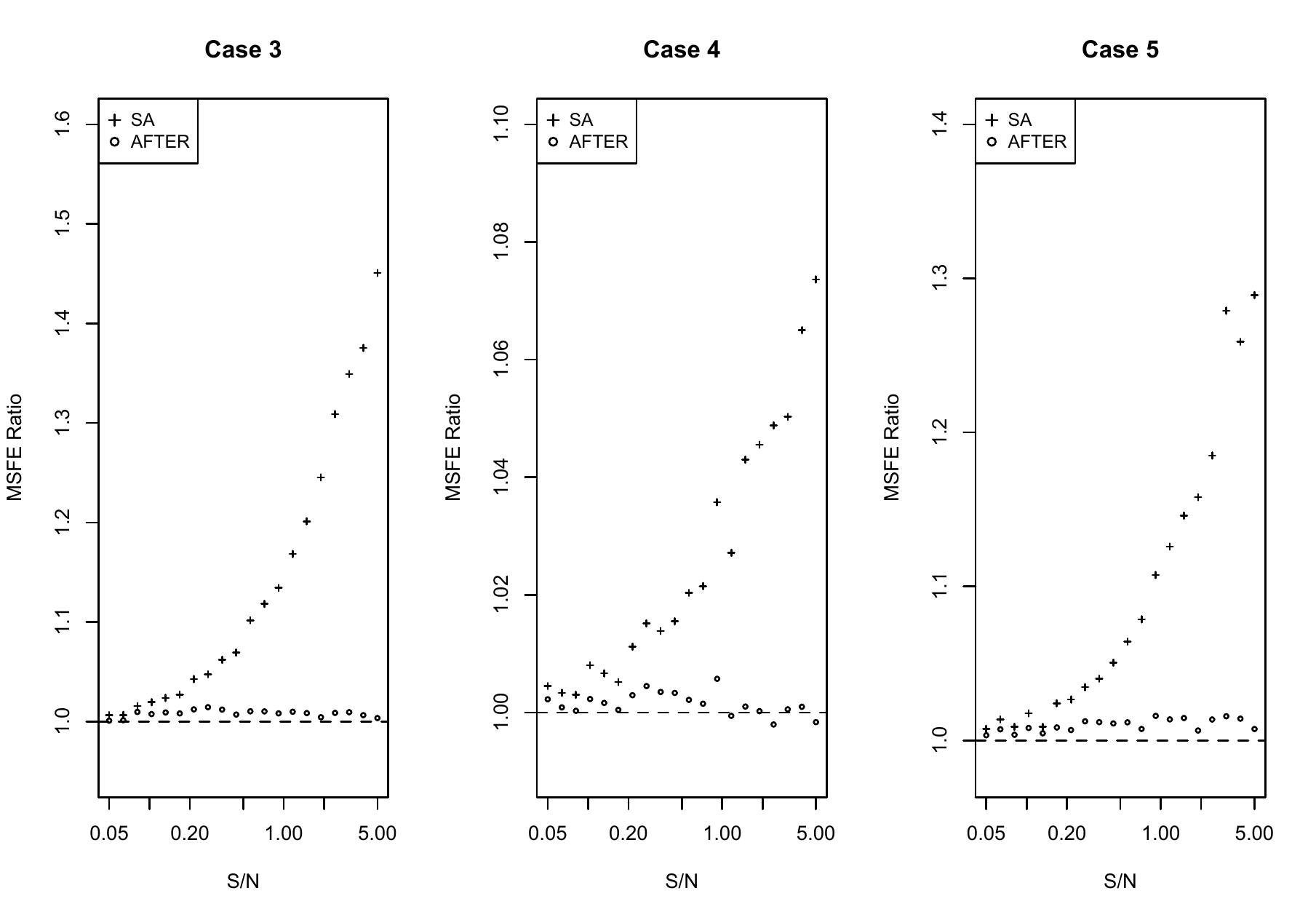}
\caption{Studying the robustness of SA against adding new
  candidate forecasts.}
\label{fig:toy3_5}
\end{figure}

\section{Improper Weighting Formulas: A Source of the FCP Revisited}
\label{subsec:breaks}

Generally speaking, the popular forecast combination methods often implicitly
assume that the time series and/or the forecast errors are stationary.  It is expected in theory
that they should perform well if we have access to long enough
historical data. In practice, however, such derived weighting formulas can often be
unsuitable when the DGP changes and the
candidate forecasts cannot adjust quickly to the {new }reality. For example,
it is often believed that {structural} breaks can
unexpectedly happen, making the relative performance of {the }candidate forecasts unstable and giving us the impression that SA
performs well. 

Next, we use a Monte Carlo example to illustrate the FCP under
{structural} breaks. Rather than assuming deterministic shift{s} in
information variable{s} \citep{hendry2004pooling}, we consider breaks in the
DGP dynamics:
\begin{equation*}
y_t=\begin{cases}
\sum_{k=1}^4\beta_{1,k}y_{t-k} + \varepsilon_t &\mbox{if }  1\leq t\leq 50, \\
 \beta_{2,1}y_{t-1}+\beta_{2,2}y_{t-2} + \varepsilon_t &\mbox{if }  51\leq t\leq 100, \\
\beta_{3,1}y_{t-1} + \varepsilon_t &\mbox{if }  101\leq t\leq 150, \\
\end{cases}
\end{equation*}
where the coefficients $\beta_{j,k}$ ($j=1,2,3$) are randomly
generated from the uniform distribution on $(0,1)$, and
$\varepsilon_t$'s are $i.i.d.$ $N(0,1)$. Here, {structural} breaks happen at $t=50$ and
$t=100$. The candidate forecast models are autoregression{s}
from lag 1 to lag 6, and we apply SA, BG, LinReg and AFTER to generate
the combined forecasts. The
simulation is repeated 100 times, and the last 100 time points serve as the
evaluation period to obtain the average forecast
risk. For comparison, we consider BG, LinReg and AFTER methods
with estimation rolling window size $rw=20$ or 40, meaning only the most recent $rw$ observations are used to estimate the weights for each forecast. The results are summarized in
Table~\ref{tab:breaks}. The average forecast risk is normalized with
respect to SA, and numbers in parentheses are standard errors.

\begin{table}[!ht]
\caption{Comparing the normalized average forecast risk of different combination methods
  under structural breaks.}
\begin{center}
\begin{tabular}{lcccc}
\toprule
& SA & LinReg &  BG & AFTER  \\
\midrule
standard & 1.000 & 1.026 (0.011) & 1.005 (0.003) & 1.047 (0.010)   \\
$rw=40$  & 1.000 &1.060 (0.033) & 0.992 (0.002) & 0.991 (0.009)   \\
$rw=20$  & 1.000 &1.64 (0.42) & 0.980 (0.003) & 0.952 (0.007)   \\
\bottomrule
\end{tabular}
\end{center}
\label{tab:breaks}
\end{table}

We can see from Table~\ref{tab:breaks} that all three standard combining
methods, when finding weights using all historical data, underperform
compared to SA due to the unstable relative performance of candidate
forecasts. As we shrink the estimation window size to
the most recent 40 and 20 time points, BG and AFTER achieve
better performance than SA while the performance of LinReg
worsens. This result can be understood by noting that there are two
opposing factors when we shrink the weight estimation window. When
using only the most recent forecasts, we decrease the bias of the weighting
formula supported by the old data but simultaneously increase the
variance of the estimated weight. Among the three methods considered,
the estimation error factor dominates for LinReg. On the other hand, AFTER is not designed to aggressively target the
optimal weight, thus benefiting the most from the shrinking rolling
window. 

Due to the complex impact of {structural} breaks on forecast
combination methods, it is arguably true
that the focus should be made on how to
detect the problem (see, e.g., \citealp{altissimo2003strong};
\citealp{davis2006structural}) and how to come up with new
 combining forms accordingly (e.g., using the
most recent observations to avoid an improper weighting formula). However, proper identification of structural breaks can be difficult to achieve in practice, and this example shows that in the presence of {structural} breaks, the relative performance of SA is not {as robust as }BG and AFTER with {na\"{\i}vely chosen }rolling windows.

\section{Linking Forecast Model Screening to FCP}
\label{subsec:screening}

In empirical studies, the candidate forecasting models are often
screened/selected in some way to generate a smaller set of candidates
for combining. As is demonstrated in Case 3 of section~\ref{sec:cocktail}, the performance of SA is particularly susceptible
to poor-performing candidate models.  The common practice of model
screening may contribute to improving the performance of SA.

Next, we illustrate the impact of screening with a Monte Carlo
example. Let $\mathbf x_t\in \real^{p}$ ($p=20$) be the $p$-dimensional information variable vector
randomly generated from {a }multivariate normal distribution with mean
$\mathbf 0$ and covariance $\Sigma$, where
$(\Sigma)_{i,j}=\rho^{\abs{i-j}}$ and $\rho=0$ or 0.5. Consider a DGP with linear model setting 
\begin{equation*}
y_t=\mathbf x_t^T\boldsymbol{\beta}+\varepsilon_t,
\end{equation*}
where coefficient  $\boldsymbol{\beta}=(3,3,2,1,1,1,1,0,0,\cdots,0)$ and
$\varepsilon_t$ are $i.i.d.$ $N(0,\sigma^2)$ with  $\sigma=2$ or 4.  Under this setting, only the
first 7 variables in $\mathbf x_t$ are important for $y_t$, while the
remaining variables are redundant. 

If we assume that the analyst has full access to the information
vector $\mathbf x_t$'s, we may build
linear models as the
candidate forecasts with any subset of the information variables. It is known from \citet{wang2014adaptive} that if we select the best subset model
with the right size using the ABC criterion \citep{yang1999model} or
combine the subset regression models by proper adaptive combining
methods \citep{yang2001adaptive}, the prediction risk can adaptively
achieve the minimax optimality over soft and hard sparse function
classes. Inspired by this result, we consider the following screening-and-combining
approach. First, given the model size (that is, the number of information
variables used in a candidate linear model), choose the best 
OLS model based on estimation mean square error. Second, from the $p$
models selected from the first step, find the top $X$\% ($X=10,20,40,60,80$) of the models
based on the ABC criterion. Note that the ABC criterion for a subset model with
size $r$ is
$ABC(r) =\sum_{t=1}^n(y_t-\hat y_{t,r})^2+2r\sigma^2+\sigma^2 \log {p \choose
r}$, where $n$ is the estimation sample size, $\hat y_{t,r}$ is
the fitted response, and $\sigma^2$ can be replaced by the estimation
mean square error.  The remaining subset models after the
two-step screening are used to
build the candidate forecasts for combining. In simulation, the total
time horizon is set to be 200. The screening procedures are applied to
the first 100 observations, and the remaining models are used to build the candidate forecasts for
the latter 100 time points. Different forecast combination methods are
applied, and their performances are evaluated
using the last 50 observations. The simulation is repeated
100 times, and the normalized average forecast risk (relative to SA) is
summarized in Table~\ref{tab:screening}.

\begin{table}[!ht]
\caption{Comparing the normalized average forecast risk of different
  forecast combination methods after the screening procedure.}
\begin{center}
\begin{tabular}{lccccc}
\toprule
Top $X$\% & 10\% & 20\% &  40\% & 60\% & 80\%  \\
\midrule
\multicolumn{6}{c}{$\sigma=2$, $\rho=0$}\\
\midrule
AFTER & 0.998 & 0.989 & 0.966 & 0.951 & 0.945  \\
BG  & 1.000 & 0.999 & 0.997 & 0.997 & 0.996   \\
LinReg  & 1.017 &1.024 & 1.056 & 1.098 & 1.151   \\
\midrule
\multicolumn{6}{c}{$\sigma=2$, $\rho=0.5$}\\
\midrule
AFTER & 0.996 & 0.990 & 0.968 & 0.956 & 0.951  \\
BG  & 1.000 & 0.998 & 0.997 & 0.997 & 0.996   \\
LinReg  & 1.013 &1.024 & 1.043 & 1.095 & 1.159   \\
\midrule
\multicolumn{6}{c}{$\sigma=4$, $\rho=0.5$}\\
\midrule
AFTER & 0.994 & 0.987 & 0.984 & 0.981 & 0.974  \\
BG  & 0.999 & 0.998 & 0.998 & 0.998 & 0.997   \\
LinReg  & 1.002 & 1.012 & 1.056 & 1.101 & 1.163   \\
\midrule
\multicolumn{6}{c}{$\sigma=4$, $\rho=0.5$}\\
\midrule
AFTER & 0.995 & 0.990 & 0.976 & 0.969 & 0.961  \\
BG  & 1.000 & 0.999 & 0.998 & 0.997 & 0.997   \\
LinReg  & 1.004 & 1.010 & 1.030 & 1.086 & 1.136   \\
\bottomrule
\end{tabular}
\end{center}
\label{tab:screening}
\end{table}

Table~\ref{tab:screening} shows that AFTER 
outperforms all the other competitors, including SA. This is
consistent with our understanding of a typical CFA scenario,
under which AFTER is the proper choice of combining methods. However, as we decrease $X$ and select smaller sets of candidate forecasts for
combining, the performance of SA gradually approaches that of AFTER.
Such a result is not entirely surprising considering that when only the top few
models are selected, simply averaging them can perform similarly to
the optimal results obtained by
the proper subset selection or combination methods 
\citep{wang2014adaptive}. LinReg, which is not a proper choice under
the CFA scenario, appears to underperform
compared to SA. 
As $X$ decreases, LinReg becomes less subject to weight estimation error, and
the performance of LinReg improves relative to SA. 

From this example, we can see that the performance of SA is not
robust to {the }degree of screening. Generally, it is a very challenging
task to ensure an optimal screening to make SA perform well.  As a
result, although SA works relatively well in this particular example
for aggressive screening (keeping very few candidates), SA should not
be preferred in general. Without a good screening/selection rule, it leaves too much
freedom for the analyst to make poor decisions. We note that a
possible solution is to first create
new candidate forecasts (e.g., forecasts generated by linear
regression methods) to utilize most or all of the important
information, and then the
roles of a good screening/selection rule can be played by applying
the multi-level AFTER approach (introduced in section~\ref{sec:cocktail0}) on both
the original forecasts and the combined forecasts to reduce the
influence of the poor-performing or redundant forecasts.

\section{Real Data Example}
\label{sec:real_data}

In this section, we study the U.S. SPF (Society of Professional Forecasters) dataset to evaluate SA and the mAFTER strategy. This dataset is a quarterly survey on macroeconomic
forecasts in the United States. \citet{lahiri2013machine} nicely
handled the missing forecasts by adopting two missing forecast
imputation strategies known as the regression imputation (REG-Imputed) and the simple
average imputation (SA-Imputed) to generate the complete panels. 
As pointed out by \citet{lahiri2013machine}, the change of data
administration agency in 1990 and the subsequently shifting missing
data pattern make it difficult to use the entire data period
for meaningful evaluation. Therefore, we inherit their missing
forecast imputation as well as the forecast selection strategies, and focus on the period from 
1968:Q4 to 1990:Q4 to evaluate the
performance of the mAFTER strategy. 

Three macroeconomic variables
are considered: {seasonally-adjusted }annual rate of change for GDP price deflator (PGDP),
growth rate of real GDP (RGDP) and quarterly
average of monthly unemployment rate (UNEMP). The datasets for RGDP
and PGDP have 14 candidate forecasts, and the datasets for UNEMP have
13 candidate forecasts. Each forecast provides
$g$-quarter ($g=1,2,3,4$) ahead forecasting. We apply SA, AFTER, BG,
LinReg and mAFTER to each SPF dataset of a 
macroeconomic variable with a given missing forecast imputation method. Each forecast combination
method uses the first
one fourth of the total time horizon to build up the initial weights, and the
remaining time points are used to calculate the normalized MSFE of
each method relative to SA. By taking the average over the four MSFEs
that correspond to the 1,2,3,4-quarter ahead forecasting, we summarize
the performance of
different combining methods in Table~\ref{tab:real_data1}.

\begin{table}[!ht]
\caption{Comparing the performance of forecast combination methods with
SPF datasets (values shown are normalized MSFEs averaged
over 1,2,3,4-quarter ahead forecasting).}
\begin{center}
\begin{tabular}{lccccc}
\toprule
Target Variable & SA & LinReg &  BG & AFTER & mAFTER \\
\midrule
\multicolumn{6}{c}{REG-imputed}\\
\midrule
PGDP & 1.00 & 1.88 & 0.95 & 0.90  & 0.90 \\
RGDP & 1.00 &1.64 & 1.00 & 1.11  & 1.01 \\
UNEMP & 1.00 & 1.79 & 0.99 & 0.98  & 0.98 \\
\midrule
\multicolumn{6}{c}{SA-imputed}\\
\midrule
PGDP & 1.00 & 2.17 & 0.98 & 0.95  & 0.95 \\
RGDP & 1.00 &1.83 & 1.00 & 1.13  & 1.03 \\
UNEMP & 1.00 & 1.69 & 0.99 & 0.97  & 0.98 \\
\bottomrule
\end{tabular}
\end{center}
\label{tab:real_data1}
\end{table}

From Table~\ref{tab:real_data1}, although AFTER performs quite differently with different target
macroeconomic variables, the mAFTER strategy delivers overall
robust performance for all three variables. For PGDP, AFTER performs the
best\deleted{,} and beats SA by as much as 10\%. Using mAFTER successfully maintains this
advantage over SA.  For RGDP, while SA and BG beat
AFTER by up to 13\%, mAFTER successfully pulls the
performance to within 3\% of SA. Finally, for the UNEMP variable, SA, BG and AFTER all
perform very similarly with no more than a 3\% difference, and the
performance of mAFTER does not deviate much from either SA or
AFTER. The LinReg method that aggressively pursues the optimal
weight performs poorly for all three target variables. It is
interesting to note from Figure~\ref{fig:real_data3} that for both PGDP and RGDP variables, the
largest performance difference between SA and AFTER {is found in }the one-quarter ahead forecasting; in each case, mAFTER
robustly matches the better of SA and AFTER. 

 \begin{figure}[!ht]
\centering
\hspace{-.4in}
\subfigure{
\includegraphics[scale=.57]{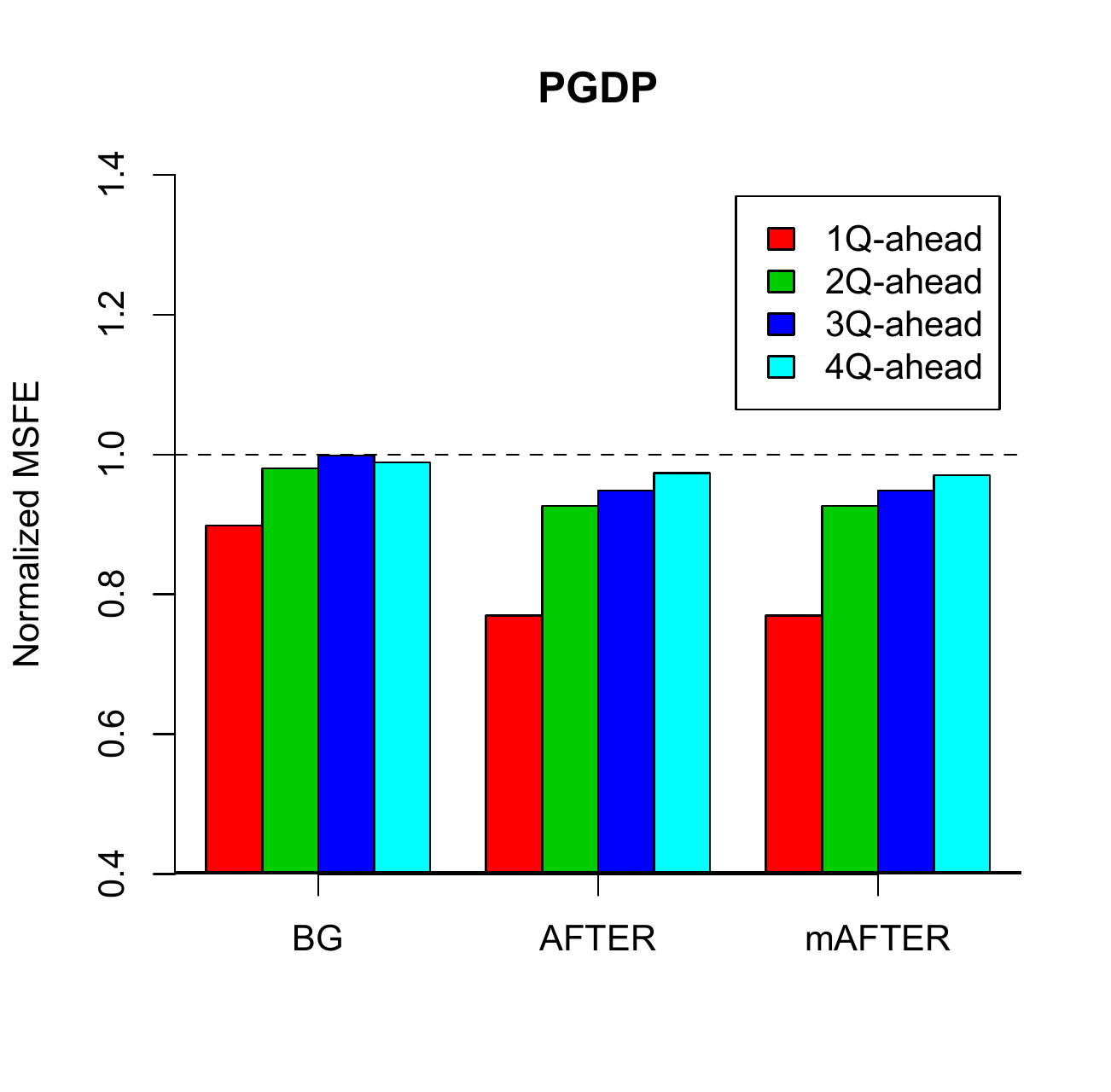}
}
\hspace{-.2in}
\subfigure{
\includegraphics[scale=.57]{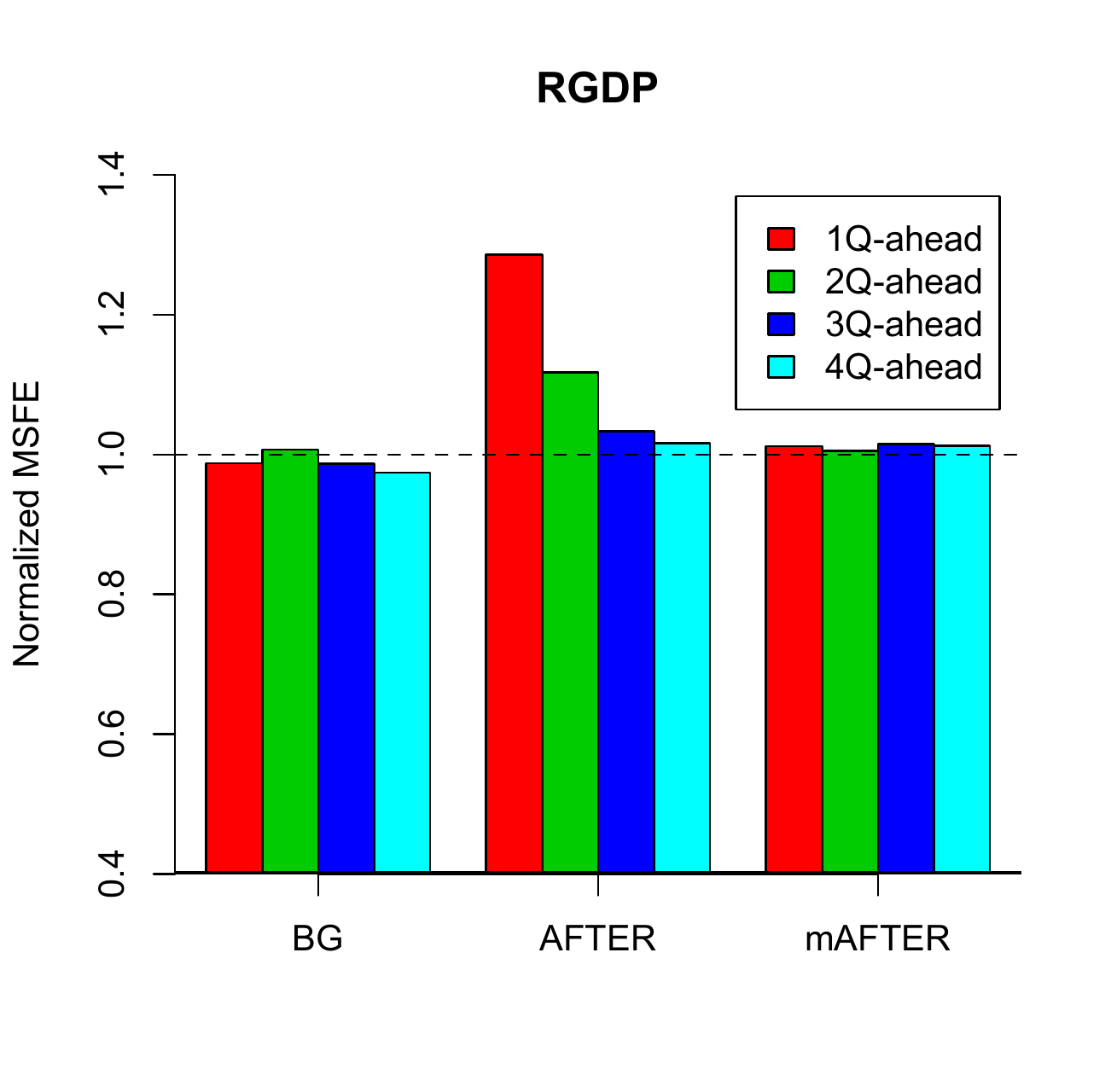}
}
\caption{Comparing normalized MSFEs of different forecast combination
  methods with REG-Imputed SPF datasets. Left panel: PGDP variable. Right
  panel: RGDP variable. For each method, the bars
  from left to right represents 1,2,3,4-quarter ahead forecasting
  results, respectively. The dashed line represents the SA baseline. }
\label{fig:real_data3}
\end{figure}

\section{Conclusions}
\label{sec:conclusion}

Inspired by the seemingly mysterious FCP, we
provide our explanations of why the puzzle often occurs and
investigate when a sophisticated combining method can work well compared to
the simple average (SA). Our study illustrates that the following
reasons can contribute to the puzzle. 

First,
estimation error is known to be an important source of FCP. Both theoretical and empirical evidence show that {a }relatively small sample size may prevent some combining
methods from reliably estimating the optimal weight. Second, FCP can
appear if we apply a combining method without
consideration of the underlying data scenarios. The relative performance of SA may depend heavily on which scenario is more proper for the data. Third, the weighting formula of the combining methods is not always appropriate for the data, because {structural} breaks and shocks can unexpectedly happen. The weighting formula obtained by
sophisticated methods cannot adjust fast enough to the reality,
resulting in performance less stable  than SA.  Fourth, candidate
forecasts are often screened in some way so that the remaining forecasts
used for combining tend to have similar performance, and SA may tend to work well
in such cases. However, SA can be sensitive to the screening process, {and }enlarging the pool of candidates may benefit other combination methods; therefore, empirical observations that SA works well {after model screening }should be taken with a grain of salt.  Fifth, there may be publication bias in that people tend to report the existence of FCP when SA gives good empirical
results but may not emphasize the performance of SA when it gives mediocre results.

Regarding the first two reasons above, our study shows that it is not hard
to find data and build candidate forecasts in a
certain way to favor a sophisticated or simple method. Under
the CFA scenario, we realize that the heavy estimation
price can be avoided by applying combining methods designed to target the
performance of the best candidate forecast. Under the CFI
scenario, although past literature has properly 
pointed out the potentially high cost of estimation error when
targeting the optimal weight, it turns out that we do not have to pay the high cost. Indeed, a carefully designed mAFTER strategy
can perform aggressively to target the optimal weight when information is sufficient to support exploiting the optimal weighting and perform conservatively like SA when {the degree of }estimation error is high. mAFTER can also intelligently perform according to the underlying
scenario (CFA or CFI), avoiding the puzzle caused by improperly choosing the combining methods.  

SA certainly can be the best or among the top combining methods, as observed empirically and reported in the literature. It may be particularly useful when one
can legitimately narrow the focus to just a few well-behaving candidate forecasts. However, since the uncertainty of the process used to
reach the small set of candidates is not reflected in the showcase examples in the literature, the ``conditional" results in favor of SA may not be replicable when one starts from scratch with inhomogeneous raw models/forecasts. For such problems, the performance of SA may span {the }whole spectrum, from terrible to on top of the chart. Also, when information is rich for a stable forecasting problem, SA may lose greatly to a model-based method (e.g., regression). In contrast, when the analyst  has little confidence {in }basic modeling assumptions on the data {or in }the quality of the available forecasts, perhaps SA (or the like) would be the choice to take.    

The repeatedly reported puzzle in literature
tends to give the sentiment that sophisticated methods are not
trustworthy and simple methods should be used. Based on our
understanding and the numerical results, it seems fair to say that if the
sophisticated methods in those studies do not perform well, it is actually because
they are not sophisticated enough, not the other way around! In
particular, when SA is considered by mAFTER as a candidate, the
possible advantage of SA is retained while the un-robustness of SA is
avoided. To a large extent, {the forecast combination puzzle }no longer exists if we are able to
move forward intelligently
by integrating the strengths of different combining methods.

\renewcommand{\theequation}{A.\arabic{equation}}
\renewcommand{\thesubsection}{\Alph{subsection}}
\renewcommand{\theassumption}{A.\arabic{assumption}}
\setcounter{equation}{0}  
  
\section*{APPENDIX}

\subsection{Assumptions of Proposition~\ref{cor:cocktail}}
\label{app:regularity}
The following two assumptions are sufficient regularity conditions for
Propostion~\ref{cor:cocktail}. Note that Assumption~\ref{ass:bound} is
satisfied if we truncate the candidate forecasts to have certain lower and
upper bounds. Assumption~\ref{ass:moments} is satisfied if the
conditional distributions of the random noise are
sub-Gaussian. 

\begin{assumption}
\label{ass:bound}
There exists a positive constant $M$ such that the candidate
forecasts satisfy with probability 1 that
\begin{equation*}
\sup_{1\leq i\leq K, 1\leq t\leq T} \abs{m_t-\hat y_{t,i}} \leq M.
\end{equation*}
\end{assumption}

\begin{assumption}
\label{ass:moments}
There exist\added{s} a constant $r_0>0$ and continuous functions
$0<h_1(r),h_2(r)<\infty$ on $[-r_0,r_0]$ such that for every $1\leq
t\leq T$ and $r\in [-r_0,r_0]$,
\begin{align*}
\mathbb E\bigl(\abs{\varepsilon_t}^2\exp(r
\abs{\varepsilon_t})|\,\mathbf x_t, \mathbf z_{t-1}\bigr)&\leq h_1(r),\\
\mathbb E\bigl(\exp(r
\abs{\varepsilon_t})|\,\mathbf x_t, \mathbf z_{t-1}\bigr)& \leq h_2(r)
\end{align*}
with probability 1.

\end{assumption}

\subsection{Propositions and Proofs}

\begin{prop}
\label{prop:case1}
Under the settings of Case 1, the average forecast risk of Forecaster 1 relative to the SA satisfies
\begin{equation*}
\frac{R_{T,1}}{R_{T,\text{SA}}} \rightarrow
\frac{\sigma^2}{\sigma^2+\beta^2\sigma_X^2/4} \quad \text{as }\,
T\rightarrow \infty.
\end{equation*}
In addition, if we consider the weight vectors in $\real^2$, the asymptotic optimal combination
weight $w^*$ satisfies
\begin{equation*}
w^*=:\arg\min_{\mathbf w\in\real^2} \left(\lim_{T\rightarrow\infty}
R_{T,\mathbf w}\right) = \left({1\atop 0}\right).
\end{equation*}
\end{prop}

\begin{prop}
\label{prop:case2}
Under the settings of Case 2, if we assume that $\beta=\beta_1=\beta_2$
and $\sigma_X=\sigma_{X_1}=\sigma_{X_2}$, the average forecast risk of
Forecast $i$ ($i=1,2$) relative to the SA satisfies
\begin{equation}
\label{eq:ratio22}
\frac{R_{T,i}}{R_{T,\text{SA}}} \rightarrow
\frac{\sigma_X^2\beta^2(1-\rho^2)+\sigma^2}{\sigma_X^2\beta^2(1-\rho^2)(1-\rho)/2+\sigma^2} \quad \text{as }\,
T\rightarrow \infty.
\end{equation}
In addition, if we further assume  $\rho=0$, the asymptotic optimal combination
weight $\tilde w^*$ under the restriction $\Theta=\{\mathbf w:
w_1+w_2=1\}$ satisfies
\begin{equation}
\label{eq:weight_improve}
\tilde w^*=:\arg\min_{\mathbf w\in\Theta} \left(\lim_{T\rightarrow\infty}
R_{T,\mathbf w}\right) = \left({1/2\atop 1/2}\right),
\end{equation}
and the asymptotic optimal combination weight $w^*$ without the
restriction satisfies
\begin{equation}
\label{eq:weight_improve1}
 w^*=:\arg\min_{\mathbf w\in\real^2} \left(\lim_{T\rightarrow\infty}
R_{T,\mathbf w}\right) = \left({1 \atop 1}\right),
\end{equation}

\end{prop}

The proof of Proposition~\ref{prop:case1} is similar to that of
Proposition~\ref{prop:case2}. In the following, we provide a sketch
for the proof of Proposition~\ref{prop:case2}.

\begin{proof}[Proof of Proposition~\ref{prop:case2}]
Let $r_{T,1}=\mathbb E(y_T-\hat
y_{T,1})^2$, $r_{T,1}=\mathbb E(y_T-\hat y_{T,2})^2$ and $r_{T,\mathbf
  w}=\mathbb E(y_T-\hat y_{T,\mathbf w})^2$ be the
point-wise forecast risks at time $T$ for forecaster 1, forecaster 2 and the
combined forecast,
respectively. We {will }first verify that under the restriction $\Theta=\{\mathbf w:
w_1+w_2=1\}$,
\begin{align*}
r_{T+1,1} &= \sigma^2\Bigl(1+\frac{1}{T-2}\Bigr)+\sigma_{X_2}^2\beta_2^2+\sigma_{X_1}^2\beta_2^2\mathbb
E\Bigl(\hat\rho^2\frac{\hat\sigma_{X_2}^2}{\hat\sigma_{X_1}^2}\Bigr)
-2\rho\sigma_{X_1}\sigma_{X_2}\beta_2^2\mathbb
E\Bigl(\hat\rho\frac{\hat\sigma_{X_2}}{\hat\sigma_{X_1}}\Bigr),\\
r_{T+1,2} &= \sigma^2\Bigl(1+\frac{1}{T-2}\Bigr)+ \sigma_{X_1}^2\beta_1^2+\sigma_{X_2}^2\beta_1^2\mathbb
E\Bigl(\hat\rho^2\frac{\hat\sigma_{X_1}^2}{\hat\sigma_{X_2}^2}\Bigr)
-2\rho\sigma_{X_1}\sigma_{X_2}\beta_1^2\mathbb
E\Bigl(\hat\rho\frac{\hat\sigma_{X_1}}{\hat\sigma_{X_2}}\Bigr),{ \text{and}}\\
r_{T+1,\mathbf w} &= \sigma^2(1-w_1^2-w_2^2)+w_1^2 r_{T+1,1} + w_2^2 r_{T+1,2} +
2w_1w_2\Bigl(\rho\sigma_{X_1}\sigma_{X_2}\beta_1\beta_2\bigl(1+\mathbb
E(\hat\rho)^2\bigr)\\
&\qquad\qquad  - \sigma_{X_1}^2\beta_1\beta_2\mathbb
E\bigl(\hat\rho\frac{\hat\sigma_{X_2}}{\hat\sigma_{X_1}}\bigr)
-\sigma_{X_2}^2\beta_1\beta_2\mathbb
E\bigl(\hat\rho\frac{\hat\sigma_{X_1}}{\hat\sigma_{X_2}}\bigr) +
\frac{\rho\sigma_{X_1}\sigma_{X_2}\sigma^2}{T}\mathbb E\bigl(\frac{\hat\rho}{\hat\sigma_{X_1}\hat\sigma_{X_2}}\bigr)\Bigr),
\end{align*}
where $\hat\sigma_{X_i}=\sqrt{\sum_{t=1}^T x_{t,i}^2/T}$ is the
estimated covariate standard deviation  ($i=1,2$) and $\hat\rho=\frac{\sum_{t=1}^T x_{t,1}x_{t,2}}{T
  \hat\sigma_{X_1}\hat\sigma_{X_2}}$ is the estimated covariate
correlation. 

{First, we have}
\begin{align*}
r_{T+1,1} &= \mathbb
E(y_{T+1}-x_{T+1,1}\hat\beta_{T+1,1})^2\\
&=\mathbb E\Bigl(\varepsilon_{T+1}+x_{T+1,1}\beta_1+x_{T+1,2}\beta_2 -
\frac{x_{T+1,1}\sum_{t=1}^Tx_{t,1}y_t}{\sum_{t=1}^Tx_{t,1}^2}\Bigr)^2\\
& = \sigma^2+ \mathbb E\Bigl(x_{T+1,1}\beta_1+x_{T+1,2}\beta_2 -
\frac{x_{T+1,1}\sum_{t=1}^Tx_{t,1}(x_{t,1}\beta_1+x_{t,2}\beta_2+\varepsilon_t)}{\sum_{t=1}^Tx_{t,1}^2}\Bigr)^2\\
&=\sigma^2+\mathbb E(x_{T+1,2}\beta_2)^2 + \mathbb
E\Bigl((x_{T+1,1}\beta_2)^2\bigl(\frac{\sum_{t=1}^Tx_{t,1}x_{t,2}}{\sum_{t=1}^Tx_{t,1}^2}\bigr)^2\Bigr)
+\mathbb
E\Bigl(\frac{x_{T+1,1}^2(\sum_{t=1}^Tx_{t,1}\varepsilon_t)^2}{(\sum_{t=1}^Tx_{t,1}^2)^2}\Bigr)\\
& \qquad - 2\,\mathbb
E\Bigl(\frac{x_{T+1,1}x_{T+1,2}\beta_2^2\sum_{t=1}^Tx_{t,1}x_{t,2}}{\sum_{t=1}^Tx_{t,1}^2}\Bigr)\\
&= \sigma^2+\sigma_{X_2}^2\beta_2^2+\sigma_{X_1}^2\beta_2^2\mathbb
E\Bigl(\hat\rho^2\frac{\hat\sigma_{X_2}^2}{\hat\sigma_{X_1}^2}\Bigr) +\frac{\sigma^2}{T-2}-2\rho\sigma_{X_1}\sigma_{X_2}\beta_2^2\mathbb
E\Bigl(\hat\rho\frac{\hat\sigma_{X_2}}{\hat\sigma_{X_1}}\Bigr).
\end{align*}
{The expression for }$r_{T+1,2}$ can be derived similarly. For $r_{T+1,\mathbf w}$,
we have
\begin{align*}
r_{T+1,\mathbf
  w} &=\mathbb E(y_{T+1}-w_1\hat y_{T+1,1}-w_2\hat y_{T+1,2})^2\\
&= \sigma^2+\mathbb
E\Bigl(w_1(x_{T+1,1}\beta_1+x_{T+1,2}\beta_2-x_{T+1,1}\hat\beta_{T+1,1})\\
& \qquad{} \qquad{} +w_2(x_{T+1,1}\beta_1+x_{T+1,2}\beta_2-x_{T+1,2}\hat\beta_{T+1,2})\Bigr)\\
&= \sigma^2(1-w_1^2-w_2^2)+ w_1^2r_{T+1,1}+w_2^2r_{T+1,2}\\
&\qquad +2w_1w_2\mathbb
E\Bigl((x_{T+1,1}\beta_1+x_{T+1,2}\beta_2-x_{T+1,1}\hat\beta_{T+1,1})\\
&\qquad{} \qquad{} \qquad{} \times (x_{T+1,1}\beta_1+x_{T+1,2}\beta_2-x_{T+1,2}\hat\beta_{T+1,2})\Bigr)\\
&=: \sigma^2(1-w_1^2-w_2^2)+ w_1^2r_{T+1,1}+w_2^2r_{T+1,2}+2w_1w_2A_1.
\end{align*}
With tedious algebra, it is not hard to show that
\begin{align*}
A_1&=\rho\sigma_{X_1}\sigma_{X_2}\beta_1\beta_2\bigl(1+\mathbb
E(\hat\rho)^2\bigr) - \sigma_{X_1}^2\beta_1\beta_2\mathbb
E\left(\hat\rho\frac{\hat\sigma_{X_2}}{\hat\sigma_{X_1}}\right) -\sigma_{X_2}^2\beta_1\beta_2\mathbb
E\left(\hat\rho\frac{\hat\sigma_{X_1}}{\hat\sigma_{X_2}}\right)\\
& \qquad{} + \frac{\rho\sigma_{X_1}\sigma_{X_2}\sigma^2}{T}\mathbb E\left(\frac{\hat\rho}{\hat\sigma_{X_1}\hat\sigma_{X_2}}\right).
\end{align*}
Together with the previous display, we verify the formula for
$r_{T+1,\mathbf w}$. The formulas \eqref{eq:ratio22} and \eqref{eq:weight_improve} can be
verified straightforwardly by noting that {the }$\mathbf
x_t$'s are normally distributed and {that }$r_{T,i}/R_{T,i}\rightarrow
1$ as $T\rightarrow\infty$ ($i=1,2$). When there is no restriction on
$\mathbf w$, $r_{T+1,\mathbf w}$ {can be derived }similarly as above{. Then, we can }show that when $\mathbf w=(1,1)^T$, $\lim_{T\rightarrow\infty}
R_{T,\mathbf w}=\sigma^2$, which implies \eqref{eq:weight_improve1}.

\end{proof}

\bibliographystyle{agsm}
\bibliography{references}{}

\end{document}